%% file: ec128-syrgkanis.tex
\documentclass[11pt,letterpaper]{article}

\usepackage{amsmath,amsfonts,graphpap,amscd,mathrsfs,graphicx,lscape}
\usepackage{epsfig,amssymb,amstext,xspace}
\usepackage{theorem}
\usepackage{algorithm,algorithmic}
\usepackage{paralist}
\usepackage{color}              
\usepackage{epsfig,subfigure}
\usepackage{bbm}

\newcommand{\E}{\mathbb{E}}

\newcommand{\rev}{\mathcal{R}}

\newtheorem{theorem}{Theorem}[section]
\newtheorem{lemma}[theorem]{Lemma}

\newtheorem{corollary}[theorem]{Corollary}

\def\squareforqed{\hbox{\rule{2.5mm}{2.5mm}}}

\def\QED{\ifmmode\squareforqed 
  \else{\nobreak\hfil   
    \penalty50                 
    \hskip1em                  
    \null                      
    \nobreak                   
    \hfil                      
    \squareforqed              
    \parfillskip=0pt           
    \finalhyphendemerits=0     
    \endgraf}                  
  \fi}

\def\blksquare{\rule{2mm}{2mm}}
\def\qedsymbol{\blksquare}
\newcommand{\bg}[1]{\medskip\noindent{\bf #1}}
\newcommand{\ed}{{\hfill\qedsymbol}\medskip}

\newenvironment{proof}{\bg{Proof : }}{\ed}

\newcommand{\R}{\ensuremath{\mathbb R}}

\newcommand{\I}{\ensuremath{\mathcal I}}

\newcommand{\M}{\ensuremath{\mathcal M}}
\newcommand{\B}{\ensuremath{\mathcal B}}

\newcommand{\comment}[1]{}
 {}

\newcommand{\junk}[1]{}

\newlength{\tmp} \newlength{\lpsx} \newlength{\lpsy} \newlength{\upsx} \newlength{\upsy}

\newcommand{\spe}{\text{\textsc{Spe}} }

\newcommand{\opt}{\text{\textsc{Opt}} }
\newcommand{\vcg}{\text{\textsc{Vcg}} }

\begin{document}

\title{Bayesian Sequential Auctions\renewcommand{\thefootnote}{\fnsymbol{footnote}}\footnote{Preliminary version appeared in EC'12. Current manuscript corrects an
error in Theorem \ref{thm:matching} (leading to a bound of $3.16$ instead of $3$) and improves the 
bound in Theorem \ref{thm:matroid} from $3$ to $2.58$.}\renewcommand{\thefootnote}{\arabic{footnote}}\setcounter{footnote}{0}}

\author{Vasilis Syrgkanis\thanks{ {\tt vasilis@cs.cornell.edu} Dept of Computer
Science, Cornell University. Supported in part by ONR grant N00014-98-1-0589 and
NSF grants CCF-0729006.}
\and
\'{E}va Tardos\thanks{ {\tt eva@cs.cornell.edu} Dept of Computer
Science, Cornell University. Supported in part by NSF grants CCF-0910940 and
CCF-0729006, ONR grant N00014-08-1-0031,  a Yahoo!~Research Alliance Grant, and
a Google Research Grant.}}

\maketitle

\begin{abstract}
In many natural settings agents participate in multiple different auctions that are not simultaneous. In such auctions, future opportunities affect strategic considerations of the players. The goal of this paper is to develop a quantitative understanding of outcomes of such sequential auctions. In earlier work (Paes Leme et al. 2012) we initiated the study of the price of anarchy in sequential auctions. We considered sequential first price auctions in the full information model, where players are aware of all future opportunities, as well as the valuation of all players. In this paper, we study efficiency in sequential auctions in the Bayesian environment, relaxing the informational assumption on the players. We focus on two environments, both studied in the full information model in Paes Leme et al. 2012, matching markets and matroid auctions. In the full information environment, a sequential first price cut auction for matroid settings is efficient. In Bayesian environments this is no longer the case, as we show using a simple example with three players. Our main result is a bound of $1+\frac{e}{e-1}\approx 2.58$ on the price of anarchy in both matroid auctions and single-value matching markets (even with correlated types) and a bound of $2\frac{e}{e-1}\approx 3.16$ for
general matching markets with independent types. To bound the price of anarchy we need to consider possible deviations at an equilibrium. In a sequential Bayesian environment the effect of deviations is more complex than in one-shot games; early bids allow others to infer information about the player's value.  We create effective deviations despite the presence of this difficulty by introducing a bluffing technique of independent interest.
\end{abstract}

\input{intro.tex}

\input{related-work.tex}

\input{prel.tex}

\input{matching.tex}

\input{matroid.tex}

\section*{Acknowledgements}We would like to thank William Sandholm for pointing to us the notes of Myerson and Reny on sequential equilibria for infinite type and action games. We would also like to thank the reviewers for the useful comments.

\bibliographystyle{abbrv}
\bibliography{sigproc}
\input{appendixA.tex}

\end{document}

%% file: intro.tex
\section{Introduction}
Auctions typically used in practice are extremely simple, not truthful, and do not even run simultaneously. The Web provides an environment where running auctions becomes easy, an environment where simplicity of design and decreased need of coordination is more important than ever before.

The most well-known auction design is the truthful and efficient VCG auction, which
requires large amount of coordination among sellers, having items available simultaneously, and
requires the sellers to agree on how to divide the revenue. The classical field of mechanism design
has a long and distinguished history. However, the mechanisms proposed are typically too complex, and
require careful coordination between the sellers. Mechanism design for either multidimensional valuations or
players with correlated types is proving to be rather challenging and leads to quite complicated mechanisms.
In contrast, mechanisms typically used in practice tend to be simple,
and do not satisfy the high standards of classical mechanism design even in simple environments. Bidders appear to
prefer simple mechanisms. Given the prevalent use of simple mechanisms,
it is important to design simple auctions with good
performance and to understand properties of auction designs used
in practice.

Several recent papers have studied properties of simple item-bidding auctions,
such as using simultaneous second price auctions for each item
\cite{Christodoulou2008,Bhawalkar}, or simultaneous first price auctions
\cite{Hassidim2011,Immorlica2005} and show that equilibria of these games have high social welfare.
Other recent papers study simple auctions used in practice and show that high social welfare can be achieved
even when players' valuations come from arbitrarily correlated distributions \cite{Lucier2011,DBLP:journals/corr/abs-1201-6429}.

In this paper we will focus on sequential auctions. The simplest, most natural and most common way to
auction items by different owners is to run individual single
item auctions (e.g. sell each item separately on eBay).
In \cite{PaesLeme2012} we initiated the study of the quality of outcomes in
sequential auctions in the full information setting. Full information is an extremely strong
assumption: buyers have to be aware of all future opportunities,
and the valuations of all bidders participating in future auctions. In contrast, in online auctions
(see \cite{Parkes} for a survey) it is assumed that players make strategic decisions without having any
information about the future. However, typical participants in auctions  have some information
about future events, and engage in strategic thinking about the upcoming
auctions. The Bayesian environment considered in this paper allows us to model
agents who are fully strategic, and have partial information about upcoming auctions. We focus on first price
auctions as the price of anarchy for second price auctions can be arbitrarily bad in most environments
already in the full information setting, including matching markets or even additive valuations.

The goal of this paper is to study sequential first price auctions in the Bayesian setting, relaxing the full information
assumption of \cite{PaesLeme2012}. Many of the results for
simultaneous auctions, such as \cite{Christodoulou2008,Bhawalkar}, extend naturally, and sometimes without
degradation of quality, to the Bayesian setting. However, such extension is more difficult in the sequential setting.
To bound the price of anarchy we need to consider possible deviations at an equilibrium. The sequential and Bayesian environment presents additional difficulties that we need to overcome. Early bids allow others to infer information about the player's value, causing the player behavior in later auctions to become correlated even when valuations are drawn independently, and correlated bidding makes it harder to prove price of anarchy bounds
- indicative of this is that Bayes-Nash price of anarchy proofs for simultaneous item bidding (e.g.\cite{Bhawalkar,Hassidim2011}) would fail in the
presence of correlated bidding.

Another difficulty that is more technical is that price of anarchy proofs typically consider possible deviations that the players have. However, in a sequential setting it is unclear how to predict player behavior outside the equilibrium path. To deal with this issue in \cite{PaesLeme2012} we consider deviations where players remain on the equilibrium path until a particular item of interest arrives, the item assigned to them in the efficient allocation. In the Bayesian setting, it is harder to identify what can be such an item of interest, as the efficient solution varies with the valuations of other players. We use a random draw of the valuation to select the item in question, and introduce a bluffing technique to de-correlate the behavior of players and relate the equilibrium outcome to the efficient solution in expectation.

We model strategic thinking  in a sequential Bayesian environment by using the notion of Perfect Bayesian Equilibrium. Perfect Bayesian Equilibrium models players updating their beliefs about the values of other participants throughout the game based on the information inferred from previous auctions. This update of beliefs is greatly affected by the information structure of the game. It depends on what information is available to participants after each round of the auction. Updates are needed even if players only observe their own outcomes. The main technical problem with Bayesian sequential games is that the
Bayesian update of beliefs is well defined only when players observe actions that are employed by some players type at equilibrium. Otherwise, the probability of observing such an action is zero and an update is not well defined. A Perfect Bayesian Equilibrium assumes that for such out of equilibrium histories of play, players can form any possible belief that is consistent with the history. A more restrictive notion is that of the Sequential Equilibrium, where each player uses a perturbed strategy, bidding uniformly at random with a small probability $\epsilon$. Such perturbation makes all bids possible on the equilibrium path, and hence an equilibrium has well-defined beliefs for each possible history of play. A sequential equilibrium is the limit of such $\epsilon$-perturbed equilibria as $\epsilon$ goes to 0. Our example of a matching auction shows that the outcome is inefficient even at a sequential equilibrium, but our bounds on the price of anarchy make no assumptions on the players' beliefs outside the equilibrium path. Hence, our upper bounds also hold for the concept of Perfect Bayesian Equilibrium.

The updates of beliefs depend on what information is available to the players after each round. Updates are needed even if players only observe their own outcomes, but releasing more information to the players allows for more informative updates. In our result on matroid auctions we assume that the winner (and possibly the winning price) is public information after each round. In particular, players deviations from equilibrium are not observed by other players, as long as they do not effect the outcome. Our result on matching auctions does not require any assumption on the information structure.

To illustrate the issues discussed above we start in Section \ref{sec:prel} by giving an example of a two-round sequential first price auction in the Bayesian setting. The example has two items on sale, and three bidders. Bidders 1 and 2 are interested in acquiring either of the two items and have equal value for them, while the 3rd bidder wants only the second item. In the full information setting, the only subgame perfect equilibrium of this game is to allocate the two items to the two bidders of larger value. We consider this game with valuations drawn from $[0,1]$ uniformly at random. Note that in the 2nd auction the third bidder has an information advantage. In this auction two bidders compete for the second item, the newly arrived bidder, and the bidder who lost the first auction. The advantage comes from the fact that losing the first auction contains information about the valuation of this bidder. This information asymmetry, and the fact that first price auction is not truthful, naturally leads to inefficiency. While this example is stated as a matching market, the same example with 3 bidders can also be thought of as matroid auction, as we will show in Section \ref{sec:matroid}.

\subsection{Our Results }
In section \ref{sec:matching} we consider matching markets, where player $i$ has valuation $v_{ij}$ for item $j$, and is interested in acquiring only one item. We will assume that there is free disposal, so his value for a set of items $J$ is $\max_{j\in J} v_{ij}$. We assume that the valuations of  different players are independent, and the valuation of player $i$ comes from a distribution $\mathcal{F}_i$. Note that the valuation of different items for the same player can be arbitrarily correlated, for example, players can have the same value for each item of interest. We prove a bound of $2\frac{e}{e-1}\approx 3.16$ on the price of anarchy of this game. This is an improvement even in the full information case, where in \cite{PaesLeme2012} we showed a bound of 4 for the price of anarchy of mixed Nash equilibria for the full information version. However, it is still a bit larger than the bound of 2 for the case of pure equilibria in the full information case. It is an interesting open problem if this increase is necessary. Our proof for pure equilibria in the full information setting relies on evaluating deviating bids for each player on the item they are assigned in the efficient outcome. This simple deviation strategy doesn't work in the Bayesian setting, as the efficient outcome depends on all the players' valuations. Our technique of selecting a good deviation to consider leads to a small increase in the bound. The result also extends to versions of the problem where items are not auctioned one-by-one, but rather groups of items are auctioned simultaneously in each sequential step.

In section \ref{sec:matroid} we consider a sequential matroid auction. The matroid structure limits the possible set of people that can get served to be a basis of that matroid. We assume that the players' valuations are drawn from a joint distribution $\mathcal{F}$, that can be arbitrarily correlated. In each step of the auction, we run a first price auction for selecting an element of a cut that doesn't contain a previous winner. This auction structure corresponds to the standard greedy algorithm for selecting the basis of maximum value. In  \cite{PaesLeme2012} we showed that in the full information setting, this auction implements the $\vcg$ outcome (achieves the same allocation and prices). In contrast, as the example mentioned above shows, the Bayesian information structure leads to inefficiency. We show that the resulting inefficiency is not unlimited by giving a bound of $1+\frac{e}{e-1}$ on the price of anarchy. An interesting example of the matroid auction is matchable subsets of a bipartite graph (a set of nodes in one side is independent, if there is a way to match them to the other side via a matching in the graph). We can think of this as a special case of the matching auction, when players have a single private value $v_i$ but each can only be assigned to a subset of items $S_i$. For such a matroid setting a more natural auction is to run the sequential item auction that we studied in the matching markets case, instead of the sequential cut auction. We show that the sequential item auction for this special case of matching markets has price of anarchy at most $1+\frac{e}{e-1}$ even when bidders valuations are correlated.

%% file: related-work.tex
\subsection{Related Work}

The study of Sequential Auctions in the economics literature was initiated by the
works of \cite{Weber2000} and \cite{Milgrom1982a} that analyzed first and second price sequential
auctions with unit-demand bidders in the incomplete
information model. They considered a setting where
items are  identical and players have the same value for all
the items available which is drawn independently from the same distribution.
Moreover, they limit their study to symmetric equilibria which,
in their setting, are always efficient.
Much of the literature of sequential auctions used the Milgrom and Weber model to study 
sequential auctions \cite{Ashenfelter1989,McAfee1993}.
Little is known for more complex valuation models in the incomplete information
case (see e.g. \cite{Krishna2009} for a detailed exposition on results
in sequential first price auctions), mainly due to the fact that for more complex valuation models or for asymmetric
or correlated settings computing the equilibrium analytically is hard. However, our price of anarchy
techniques allow us to expose meaningful properties for much more complex environments without the need
to explicitly calculate the equilibrium strategies.

A few papers consider the complete information case for more complex settings.
For example, \cite{Gale2001,Rodriguez2009,Bae2009} study models of multi-unit
demands but only for the case of two bidders. Most of this work depends heavily on
having just two bidders since then the subgame perfect equilibrium that survives
elimination of weakly dominated strategies is unique. In \cite{PaesLeme2012}
we initiated the study of the price of anarchy in sequential auctions with multi-unit demands
and more than two bidders. In this work we manage to
give efficiency guarantees for sequential auctions not only among more than two bidders
but also in the incomplete information case, even for asymmetric players and in some
cases even for correlated bidders.

Recent work from the Algorithmic Game Theory community studied the outcomes of
simple mechanisms for multi-item auctions.
Christodoulou, Kovacs and Schapira \cite{Christodoulou2008} and
\cite{Bhawalkar} study the case of running
simultaneous second price item auctions for combinatorial auction settings. \cite{Christodoulou2008} prove that for bidders with submodular
valuations and incomplete information, the Bayes-Nash Price of Anarchy is $2$.
\cite{Bhawalkar} study the more general case of
bidders with subadditive valuations and show that under complete information, the Price of
Anarchy of any Pure Nash Equilibrium is $2$ and under incomplete information the
Price of Anarchy of any Bayes-Nash Equilibrium is at most logarithmic in the
number of items. \cite{Hassidim2011}  and \cite{Immorlica2005} study the case of simultaneous first price auctions and show that the set of pure Nash equilibria of the game correspond to exactly the Walrasian equilibria. \cite{Hassidim2011}  also show that mixed Bayes-Nash equilibria have a price of anarchy of 2 for submodular bidders and logarithmic, in the number of items, for subadditive valuations.

The most recent related work on Matroid Auctions is that of \cite{Bikhch2010} who propose a centralized
ascending auction for selling a base of a Matroid that results in the \vcg
outcome. We study a sequential version of this matroid base auction which potentially renders
the auction more distributed. Several other recent works have focused on
mechanism design questions in the case when the feasible set of allocations
has a matroid structure \cite{Chawla2010,Hartline2009}.

%% file: prel.tex
\section{An Illustrative Example and Preliminaries}
\label{sec:prel}

We start our discussion with a simple illustrative example. Assume two items will be available for auction. First item 1 and then item 2. There are three buyers $a,b,c$. Two of the buyers $a$ and $b$ would like either of the two items, while bidder $c$ is only interested in the second item. Suppose that all players valuations are drawn from the uniform distribution $U(0,1)$ (and players $a$ and $b$ value either item the same). The following theorem, whose proof we defer to Appendix A describes the Sequential Equilibrium bidding strategy of the first stage auction of this game.

\begin{theorem}
\label{lm:example}
There is a sequential equilibrium of this game, where in the first auction, players $a$ and $b$ bid using the bidding function $b(v) = 1-\frac{\ln(1+v)}{v}$, and the players' beliefs about a player losing to a bid $b(v)$ is that the losing player has valuation uniformly distributed in the range $[0,v]$.
\end{theorem}
Observe  that players
bid more conservatively in the first auction than in the standard first price auction without accounting for their expected utility from the second auction. In Figure \ref{fig:bid_eq} we show the equilibrium bidding function of the first auction as compared to the (myopic) equilibrium bidding function of $b(v)=v/2$ not taking into account the second round.
\begin{figure}
\centering
\subfigure{
\includegraphics[height=123pt]{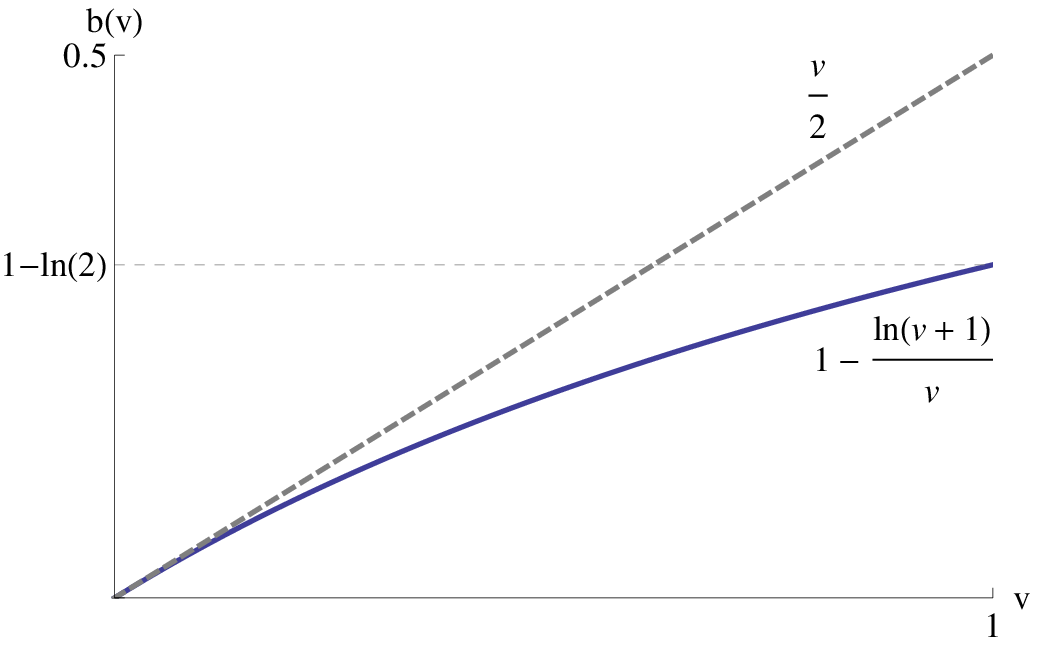}}
\subfigure{
\includegraphics[height=123pt]{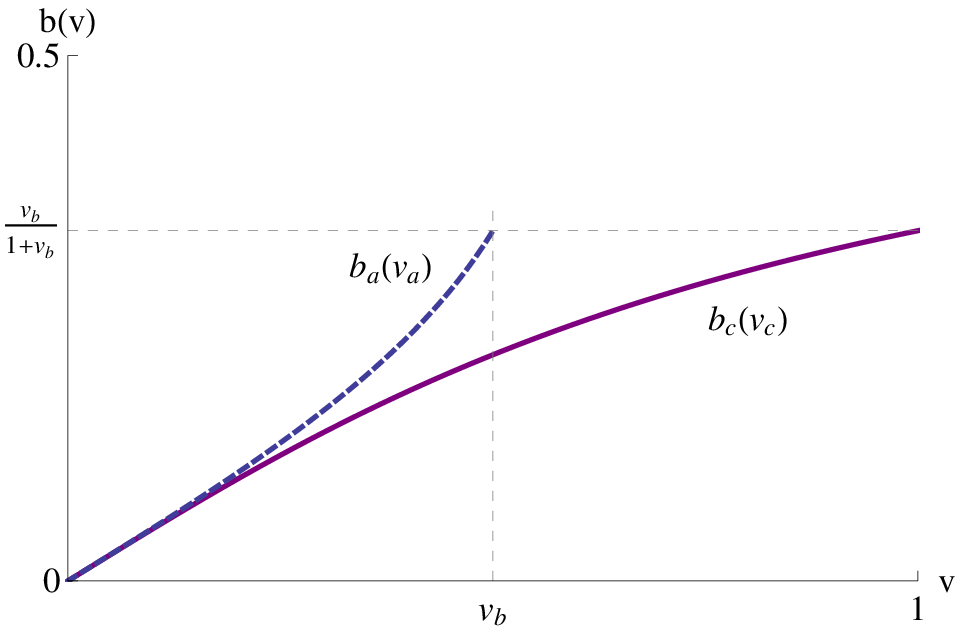}}
\caption{The left figure shows the symmetric equilibrium bidding function of the first stage auction compared to myopic bidding $b(v)=v/2$. The right figure shows the asymmetric bidding equilibrium
of the second stage auction when player $b$ wins the first auction and hence players $a,c$ compete in the second auction. The value of player $b$ is thus revealed and the belief on player $a$ is
updated to be uniform on $[0,v_b]$.}
\label{fig:bid_eq}
\end{figure}

Our main point for presenting this example is not the particular form of the equilibrium, but rather
the fact that it leads to an asymmetric auction in the second stage. Equilibria in such asymmetric auctions can be computed analytically (see e.g. \cite{Krishna2009}) and it is well known that asymmetric first price auctions lead to inefficient outcomes, i.e. there are regions of players values where the player with the smaller valuation wins. In contrast, in a full information game this simple auction has a unique equilibrium, that is
socially optimal.

\begin{corollary} \label{lm:ineff}
The above equilibrium is not guaranteed to be efficient.
\end{corollary}

The loss of efficiency comes from a factor that is necessarily introduced by our incomplete information setting: Asymmetry.
Even if players start with symmetric priors, and even if we keep the prices secret, the fact that players participate in different auctions causes asymmetry. Suppose that players started with a symmetric prior distribution $F$. Then in the second auction the loser of the first auction has a distribution that is $F$ conditional on his value being smaller than the winner of the first auction. It is well-known that in such asymmetric settings the player with the weaker distribution is more aggressive (\cite{Maskin2011,Kirkegaard2009,Kirkegaard2011}), which leads to inefficiency.


Next we formally define an equilibrium of a sequential game. In the example above, it is natural to assume that when player $a$ loses to a bid $w$, and the bidders employ a monotone bidding strategy $b(x)$ at
equilibrium, then the common belief on player $a$'s valuation is updated to condition $a$'s distribution on $v_a\le b^{-1}(w)$, since he must have bid less than $w$ if he lost. This is the natural Bayesian update of the information about the value, given that the players observed $a$ losing the first auction. Note, however that in the equilibrium claimed above the range of bids in the first auction is $[0,1-\ln{2}]$. To complete the definition of an equilibrium, we also have to define how beliefs are updated if players bid outside the range of the equilibrium, i.e. if they go off the equilibrium path.

\subsection{Extensive-form Games of Incomplete Information}
Defining an equilibrium in a sequential (or extensive form) game is more complex than in the full information setting.  In a full information game the subgame perfect equilibrium is defined naturally via backwards induction. In each stage the player's strategies need to form an equilibrium, given the expected outcome defined in future games. In a Bayesian game we need to not only describe the strategies of the players, but also define how beliefs about player values are updated given the information available at each stage, and these strategies and belief updates need to be consistent.

A \textbf{Bayesian extensive form game} is defined in our setting as follows (see
\cite{fudenberg91} for a more comprehensive treatment).
At any point in the game, player's information correspond to pairs $(v_i,h_a)$
of the player's valuation (vector) $v_i$ and the history $h_a$ of outcomes in auctions
before auction $a$. In the games we consider,  we will assume that the same information is available to all players,
hence the information depends only on the stage of the auction, which simplifies notation.  A strategy of a player is a bid for each information set. Thus a player's strategy consists of bids $b_i(v_i,h_a)$.

We can model a sequential game via a game tree, where the starting nodes correspond to different possible valuation profiles, and later nodes correspond to such valuation and a history of play so far. However, players can only base their updates on information available to them, so a player's information set consists of nodes for each possible valuation of the rest of
the players $v_{-i}$. Player $i$ doesn't have deterministic information as to which
of these nodes is implemented, since the same history of play $h_a$ could be implemented
with different valuation profiles. However, the history of play signals incomplete information about the probability distribution of valuations. This is captured by the belief system of the sequential
equilibrium. Specifically in a sequential equilibrium a player has to have beliefs
over the nodes of an information set that are \text{consistent with the strategies}.
This forward consistency is what creates the complexity of the incomplete information
setting that we will briefly describe.
Given the beliefs of the players at an information set, his bid has to maximize
his expected utility. Moreover, the belief of a player at an information set is
the Bayesian update of his initial beliefs given the history of play.

An additional difficulty arises when we consider examples, such as the three bidder example above,
when the proposed bidding strategy has limited range. In order to fully describe an equilibrium,
we need to also define how beliefs will be updated if players deviate by bidding above the range, or more generally use strategies that are not
on an equilibrium path for some instance of the players valuations. For histories that have zero density in equilibrium, Bayesian update is not well defined. However, we need to define how  beliefs will be updated to be able to evaluate if players would benefit from deviating to such strategies.

The notion of \textbf{Perfect Bayesian Equilibrium} allows for almost any updates in this case, as long as the system of beliefs is consistent with prior and the history. The \textbf{Sequential Equilibrium}, is a more refined notion. In this equilibrium
the beliefs on the information sets that are never reached by some equilibrium path
have to be limits of Bayesian updates for $\epsilon$-perturbed strategies of the
rest of the players, i.e. if the players were playing every strategy with at least
some $\epsilon$ density, then the Bayesian update of the beliefs is well-defined
for all information sets. Then for a belief to be consistent with a sequential
equilibrium it has to be the limit of the beliefs of some sequence of such $\epsilon$
strategies.

Sequential Equilibrium has been mostly used in games with finite action and type spaces. One can think of a discretized version of our auction games, where types and bids are multiples of some $\epsilon$ and then take $\epsilon$ to zero. This is very natural in auction settings since, both types and valuations are multiples of pennies. However, we could also use a recent formal definition of Sequential Equilibrium for infinite type and action space games proposed by Myerson and Reny \cite{Myerson2011}. Our efficiency results are very robust and extend to all these different type of equilibrium definitions.

\subsection{Efficiency in Auction Games}

The core property satisfied by the above solution concepts that will be needed to prove the efficiency guarantees is that a player's expected utility conditional on his value is maximized at equilibrium. Specifically, consider the normal form representation of the extensive form game defined by our auction game. A strategy of a player $i$ in this normal form representation is a function $b_i(t_i)$ that takes a player's type $t$ as input and outputs a whole contingency plan on what a player will bid at each auction for each possible history of play.
If a strategy profile $b(t)=(b_i(t_i))_{i\in N}$ constitutes an equilibrium then it must
satisfy that for any player $i$ and type $t_i$:
\begin{equation}
u_i(t_i)=\E_{t_{-i}}[u_i(b(t)]\geq \E_{t_{-i}}[u_i(b_i',b_{-i}(t_{-i})]
\end{equation}
In other words we need only that the equilibrium is a Bayes-Nash Equilibrium (BNE) of the normal form representation of our extensive form games. The Perfect Bayesian and Sequential Equilibrium concepts introduce further refinements on how the beliefs of the players should be updated according to the history of play observed and how the players behave when some non-equilibrium history is observed. 

The social welfare $SW(b)$ of a bid profile $b$ in our auction games is the sum of the utilities of the players and the auctioneer, which boils down to the total value of the players from the resulting allocation.
Given a type profile $t$, we denote with $\opt(t)$ the allocation that maximizes social welfare.

We quantify the inefficiency of an equilibrium using the concept of the Bayes-Nash Price of Anarchy which is the fraction of the expected social welfare at the worst equilibrium concept that satisfies the above condition over the expected optimal social welfare.
\begin{equation}
\text{Bayes-PoA}=\sup_{b(\cdot)\text{ is BNE}}\frac{\E_t[\opt(t)]}{\E_{t}[SW(b(t)]}
\end{equation}

%% file: matching.tex
\section{Sequential Auctions for Matching Markets}
\label{sec:matching}
In this section we consider the first example of a sequential auction. There are a set of $m$ items on sale, each with a different seller. Items can go on sale simultaneously in groups, or individually. We assume that buyers are unit-demand, we use $v_{ij}$ to denote player $i$'s value for item $j$ and assume the value of player $i$ for a set $S$ of items is his maximum value item in the set: $v_i(S)=\max_{i\in S}v_{ij}$ (i.e. we assume free disposal). We will use $v_i \in \R^{[m]}$ to denote the vector of valuations of buyer $i$. We study such markets in incomplete information setting. We assume that each player's valuation vector is drawn independently from some prior distribution $\mathcal{F}_i$ that is common knowledge. Note that a single user's valuations can be arbitrarily correlated, we only assume independence between the different buyers.

\subsection{Bayes-Nash Price of Anarchy}

In this section we prove that there is little inefficiency in the matching market setting we described. The proof is based on a ``bluffing'' technique. Bluffing here corresponds to a player with value vector $v_i$ following a strategy as if he/she had a different valuation $v'_i$ till a certain point in the auction. Such bluffing is useful, as it allows the player to use the equilibrium properties to predict prices so far, and possibly allows to take items cheaply from a different branch of the computation. Specifically, suppose that a player sets his mind on getting some item $j$. Then he can try
to get that item by pretending to be some other type in all auctions previous to item $j$. This way he has the option to reach item $j$ along several
different equilibrium paths and for some of it the price on item $j$ might be small enough that it would be better for player $i$ to take it. However, equilibrium
conditions state that none of these paths are preferable and in particular taking expectation over all such possible paths is not preferable.
This bluffing is essential in the proof mainly because the bidding behavior of each player depends on the history of play and subsequently on the whole valuation vector. Thus the maximum other bid that a player faces at each auction is an implicit function of his own valuation. Such an effect didn't arise in simultaneous item auctions with independent bidders.

We first state the theorem in the simplest context of auctioning one item each step, then we show that the theorem extends to more general settings.
\begin{theorem}\label{thm:matching}
The Bayes-Nash Price of Anarchy of a Sequential First Price Auction with unit-demand bidders is at most $2\frac{e}{e-1}\approx 3.16$, that is, the expected social welfare in the Sequential First Price Auction is at least a $1/3.16$ fraction of the social value of the efficient allocation.
\end{theorem}
\begin{proof}
Consider an instantiation of players values $v$. Let $\opt_v$ be an efficient matching allocation
for the instance $v$ and $j^*_i(v)$ be the item that player $i$ is matched to. So the allocation is maximizing $\sum_i v_{ij^*_i(v)}$. For simplicity of presentation, we will assume a pure equilibrium, that is, assume that play is a deterministic function of the type of the player and the history so far.  The result naturally extends to mixed equilibria as well.

Also let $j_i(v)$
be the maximum value item that player $i$ gets at the Perfect Bayesian Equilibrium $\spe$. We will also denote with $U_\spe$, $SW_\spe$ and
$\rev_\spe$ the expected player utility, social welfare and revenue respectively at $\spe$. In addition let
$SW_\opt$ denote the expected optimal social welfare.

To bound the price of anarchy, we will consider deviations for each player $i$. If they have a deviation that creates high utility for them, they must have this much utility also in the equilibrium outcome. In the full information setting \cite{PaesLeme2012} use a deviation where buyer $i$ attempts to get $j^*_i(v)$. This is not an option in the Bayesian setting, since $j^*_i(v)$ depends on the valuation of all the players, and hence it is not known to $i$. We consider the following randomized version $b_i'$ of the deviation for player $i$:
he draws a random sample of a valuation profile
$w$ (including his own type), identifies $j^*_i(v_i,w_{-i})$, and the goal of the deviation will be to get this item cheaply. To do this, the player $i$ will ``bluff'' and play as in the $\spe$ had he been of type $w_i$ up until item $j=j^*_i(v_i,w_{-i})$.
Then he submits a randomized bid $t$ that follows the density function $f(t)=\frac{1}{v_{ij}-t}$ with support
$[0,(1-\frac{1}{e})v_{ij}]$. 

Let $p^{-i}_j(v)=\max_{k\neq i}b_{kj}(v_k,h_j(v))$, where $h_j(v)$
is the history of play at $\spe$ up until item $j$ when players have type $v$. Denote with $b_i'(w,t)$ the instantiation
of the above mixed deviation for some random sample $w$ and a random bid sample $t$.

If $p^{-i}_j(w_i,v_{-i})<t$ then player $i$'s utility from the deviation is at least $v_{ij}-t-P_{ij}^-(w_i,v_{-i})$.
Therefore for a fixed random sample $w$ player $i$'s expected utility is at least 
$$\int_{p^{-i}_j(w_i,v_{-i})}^{(1-\frac{1}{e})v_{ij}}(v_{ij}-t)f(t)dt - P_{ij}^-(w_i,v_{-i})
=(1-\frac{1}{e})v_{ij}-p^{-i}_j(w_i,v_{-i})-P_{ij}^-(w_i,v_{-i})$$ 

Observe in this formula that we used $p^{-i}_{j}(w_i,v_{-i})$ instead of $p^{-i}_{j}(v_i,v_{-i})$. This is because player $i$ pretended to be of type $w_i$ and hence the history that
the rest of the players see up until item $j$ is as if player $i$ had type $w_{i}$.
Moreover, for an arbitrary type profile $v$, $p_j^{-i}(v)$ depends on $v_i$ only 
through the history up until item $j$.

From the above reasoning we get:
\begin{align*}
\E_t[u_i(b_i'(w,t),b_{-i}(v_{-i}))~\geq~ & (1-\frac{1}{e})v_{ij}-p^{-i}_j(w_i,v_{-i})-P_{ij}^-(w_i,v_{-i})\\
~\geq~& (1-\frac{1}{e})v_{ij}-p^{-i}_j(w_i,v_{-i})-P_i(w_i,v_{-i})
\end{align*}
where $P_i(v)$ is the total price that player $i$ pays at $\spe$ when the valuation profile is $v$.

Now taking expectation over all possible random samples, we can lower bound player $i$'s expected
utility from deviating to $b_i'$:
\begin{align*}
u_i(b_i',b_{-i}(v_{-i})) ~=~& \E_w\E_t[u_i(b_i'(w,t),b_{-i}(v_{-i}))]\\
 ~\geq~&
\E_w\left[(1-\frac{1}{e})v_{ij^*_i(v_i,w_{-i})}-p^{-i}_{j^*_i(v_i,w_{-i})}(w_i,v_{-i})-P_i(w_i,v_{-i})\right]
\end{align*}
By the equilibrium condition we have:
\begin{align*}
u_i(v_i)~\geq~& \E_{v_{-i}}[u_i(b_i',b_{-i}(v_{-i}))]\\
~\geq~& \E_{v_{-i}}\E_w\left[(1-\frac{1}{e})v_{ij^*_i(v_i,w_{-i})}-p^{-i}_{j^*_i(v_i,w_{-i})}(w_i,v_{-i})-P_i(w_i,v_{-i})\right]
\end{align*}
We work separately for each part on the right hand side:
\begin{align*}
\E_{v_{-i}}\E_w\left[v_{ij^*_i(v_i,w_{-i})}\right] = \E_{v_{-i}}\E_{w_i} \E_{w_{-i}}\left[v_{ij^*_i(v_i,w_{-i})}\right]
=\E_{w_{-i}}\left[v_{ij^*_i(v_i,w_{-i})}\right] = \E_{v_{-i}}\left[v_{ij^*_i(v)}\right]
\end{align*}
where  we used that  $w_{-i}$ is independent of $w_i$ for the $\E_{w_i}$ to be dropped, and the last equation on the first line is just a change of variables.
\begin{align*}
\E_{v_{-i}}\E_w\left[P_i(w_i,v_{-i})\right]=~&\E_{v_{-i}}\E_{w_i}\E_{w_{-i}}\left[P_i(w_i,v_{-i})\right]
=\E_{v_{-i}}\E_{w_i}\left[P_i(w_i,v_{-i})\right]=\E_{v}\left[P_i(v)\right]
\end{align*}

For the last term we will need a notation for the price of an item at some valuation profile: $p_j(v)=\max_{k} b_{kj}(v_k,h_j(v))$.
Obviously for any instance $v$ and for any player $i$: $p_j^{-i}(v)=\max_{k\neq i}b_{kj}(v_k,h_j(v))\leq \max_{k} b_{kj}(v_k,h_j(v))= p_j(v)$.

\begin{align*}
\E_{v_{-i}}\E_w\left[p_{j^*_i(v_i,w_{-i})}^{-i}(w_i,v_{-i})\right]=~&\E_{v_{-i}}\E_{w_i}\E_{w_{-i}}\left[p_{j^*_i(v_i,w_{-i})}^{-i}(w_i,v_{-i})\right]\\
=~&\E_{w_{-i}}\E_{w_i}\E_{v_{-i}}\left[p_{j^*_i(v_i,v_{-i})}^{-i}(w_i,w_{-i})\right]\\
=~&\E_{v_{-i}}\E_{w}\left[p_{j^*_i(v)}^{-i}(w)\right]
\leq~ \E_{v_{-i}}\E_{w}\left[p_{j^*_i(v)}(w)\right]
\end{align*}
Where the second line is just a convenient change of variable names. In the third equation we used independence so as to be able to consider $(w_i,v_{-i})$ as a random sample $w$ of a valuation vector.
We point out that the bluffing technique is essential in this latter computation since otherwise we would end up with a term of the form:
$\E_{v_{-i}}\E_{w_{-i}}\left[p_{j^*_i(v_i,w_{-i})}^{-i}(v_i,v_{-i})\right]$. This would lead to a term
$\E_{v_{-i}}\E_{w}\left[p_{j^*_i(v_i,w_{-i})}(v)\right]$ which when taking expectation over $v_i$ and summing among all players
would not correspond to the expected revenue at equilibrium. The fact that the players valuation affects both the history
of play and the index of the item that he is aiming at creates a problem. The bluffing technique un-correlates this dependence
since now the players value doesn't affect the history of play, but rather a random sample of his value does.

Combining the three latter equations we proved that:
\begin{equation}
u_i(v_i)\geq \E_{v_{-i}}\left[(1-\frac{1}{e})v_{ij^*_i(v)}\right]-\E_{v_{-i}}\E_{w}\left[p_{j^*_i(v)}(w)\right]-\E_{v}\left[P_i(v)\right]
\end{equation}

Taking expectation over $v_i$ for each player and adding for all players $i$ we get:
\begin{align*}
U_\spe = \sum_i u_i \geq &\sum_i \E_v\left[(1-\frac{1}{e})v_{ij^*_i(v)}\right]-\E_{v}\E_{w}\left[p_{j^*_i(v)}(w)\right]-\E_{v}\left[P_i(v)\right]\\
=& (1-\frac{1}{e})SW_\opt - \E_v\E_w\left[ \sum_{j \in [m]} p_j(w) \right]- \rev_\spe\\
=& \frac{SW_\opt}{2} - \rev_\spe - \rev_\spe
\end{align*}
where the second line follows as the sum over elements $p_{j^*_i(v)}(w)$ is the sum over all elements $p_j(v)(w)$ as the optimal assignment only effects the order of summation.
By rearranging, we get:
\begin{align*}
SW_\spe \geq (1-\frac{1}{e})SW_\opt-\rev_\spe \geq  (1-\frac{1}{e})SW_\opt-SW_\spe\implies SW_\spe \geq \frac{1}{2}(1-\frac{1}{e})SW_\opt
\end{align*}
\end{proof}

It is easy to generalize the latter proof to the case of mixed Bayes-Nash Equilibria, games where groups of items are sold simultaneously. Note that the bound of $3.16$ of this result is an improvement over the bound of $4$ for mixed Nash Equilibria of \cite{PaesLeme2012} even in the full information setting.

\begin{corollary}
The mixed Bayes-Nash Price of Anarchy of a Sequential First Price Auction where at each stage a set of items are on sale simultaneously and
bidders are unit-demand is at most $3.16$.
\end{corollary}

Further, we can generalize the above theorems
to approximately unit-demand environments in the following sense: given any instance of values $v$,
the optimal matching allocation $\opt_M$ is at least $1/\gamma$ the optimal allocation $\opt$.
Then the latter two theorems carry over with a blow-up of $\gamma$.

%% file: matroid.tex
\section{Matroid Auctions}
\label{sec:matroid}
In this section we consider a setting when $N$ players each would like to get a service. However, due the constraints of the service only a subset of them can be selected. We will assume that the constraints have a matroid structure.
Suppose there is a set of $N$ bidders and there is a matroid $\M=(N,\I)$, where the independent sets are sets that can be simultaneously served. Suppose each player $i$ has a value $v_i$ for getting the service. We assume that the valuation profile $v$ of
the bidders is drawn from a possibly correlated distribution $F$ that is common knowledge. The goal of the mechanism is to find an independent set $I$ of high value $\sum_{i \in I} v_i$ to serve.

The simple auction format will be based on the natural greedy algorithm for finding a large value independent set. A maximum value independent set is naturally a maximal independent set, a basis $\B$ of the matroid. A cut of a matroid is a set $S\subset N$ that intersects all bases. For example, in the matroid of a graph, where an edge-set is independent if it contain no cycles, basis are the spanning trees, and cuts are the set of edges crossing a graph cut.

We consider the following class of simple sequential first-price auctions: At each time step
the auctioneer picks a cut of the matroid that doesn't intersect previous winners and
runs a first price sealed-bid auction among the bidders in the cut. At the end of each auction $j$
the winner $w_j$ and the price that he paid $p_j$ are announced.
Such an auction defines an
extensive form game of incomplete information. If in each auctions bidders would announce their true value $v_i$ as bids, the auction corresponds to the standard greedy algorithm for matroids, and would result in picking the basis of maximum value. For many matroid structures cuts can be small, and hence the auction lends itself to a more distributed implementation, when only elements of a cut have to participate in one auction. The sequential structure of the auction allows participants to infer information about other players, and leads to strategies that possibly will not lead to an efficient outcome.

\subsection{Inefficiency of the Incomplete Case}
In \cite{PaesLeme2012} we studied the complete information version of this matroid auction game
and showed that it implements the \vcg outcome in the sense that all equilibria select an efficient outcome,
and make all winners pay their \vcg price.

However, when we move to the incomplete
information setting the auction no longer implements the \vcg outcome, and
inefficiency could arise at the sequential equilibria of the game.

As an example, consider the matroid on three players, where any two are independent, but not all three.
This is the graphical matroid of a triangle graph.  Suppose that all players valuations are drawn from the uniform
distribution $U(0,1)$. If we use the cut $\{a,b\}$ as the first cut, then in the next auction, the loser of the first auction is bidding against the 3rd player $c$. This is exactly the example of Lemma \ref{lm:ineff} cast as a matroid auction: players $a$ and $b$ participate in the first auction, and then $c$ plays against the loser. As we showed the result is not guaranteed to be efficient.


\subsection{Bayes-Nash Price of Anarchy}
We first present a lemma from \cite{PaesLeme2012} that will be very useful in our
main theorem. We define the notion of the
\textbf{participation graph} $\mathcal{P}(B)$ of a base $B$ to be a bipartite
graph between the elements in the base and the auctions that took place. An edge
exists between an element of the base and an auction if that element
participated in the auction.

\begin{lemma}[\cite{PaesLeme2012}] For any base $B$ of the matroid $\mathcal{M}$, $\mathcal{P}(B)$
contains a perfect matching.
\label{lem:exists_perf_match}
\end{lemma}

\begin{theorem}\label{thm:matroid}
The Bayes-Nash Price of anarchy of the Sequential First Price Cut Auction is at most $1+\frac{e}{e-1}\approx 2.58$ even if valuations are arbitrarily correlated.
\end{theorem}
\begin{proof}
Given an instantiation of the value vector $v$, let $\opt(v)$ be the optimal base under
value vector $v$ and $\spe(v)$ the base produced at a perfect Bayesian equilibrium.
From Lemma \ref{lem:exists_perf_match} we have that for any instantiation
of values $v$ there is a perfect matching between the auctions that take place and the
players in $\opt(v)$. For a player $i\in \opt(v)$ let $A_v(i)$ be the index of the auction that player $i$
is matched with in the latter matching. Also let $x_i(v)=\mathbf{1}_{i\in \spe(v)}$ and $y_i(v)=\mathbf{1}_{i\in \opt(v)}$.

Consider an instance of players' values $v$ and a player $i\in \opt(v)-\spe(v)$.
Suppose that $i$ bids according to the following randomized strategy $b_i'$: he
picks a random bid $t$ with density $f(t)=\frac{1}{v_i-t}$ and support $[0,(1-\frac{1}{e})v_i]$
and then bids $t$ in all the auctions until he either gets allocated or the auctions finish. Since $i\notin \spe(v)$ every time that
he loses an auction $j$ he doesn't affect the outcome $(w_j,p_j)$ that is announced and
therefore the change is transparent to the rest of the players. Hence, they continue using
their equilibrium strategies in the subsequent auctions and the prices at subsequent auctions are
the same as before the deviation. Let $b_i'(t)$ be the deterministic deviation for some instance
of the random bid $t$.

If $p_{A_v(i)}<t$ then player $i$ will definitely win some auction in the deviation $b_i'(t)$ and
his utility will be $v_i-t$. An important point for the above to hold is that as long as a player
is not winning then, since by assumption he wasn't winning previously too, he doesn't change the
history of play since the same price and winner announcement is made after each auction. Hence, when
$A_v(i)$ arrives, either player $i$ has already won an item or the price of $A_v(i)$ is going to be the 
equilibrium price $p_{A_v(i)}$ of that item.  Thus player $i$'s expected utility from the deviation is at least:
\begin{equation}
u_i(b_i',b_{-i}(v_{-i}))\geq y_i(v)\int_{p_{A_v(i)}}^{(1-\frac{1}{e})v_i}(v_i-t)f(t)dt =
(1-\frac{1}{e})v_iy_i(v) - p_{A_v(i)}y_i(v)
\end{equation}

If $i\in \opt(v)\cap \spe(v)$, i.e. $x_i(v)=1$, then we cannot consider the latter deviation since
$i$ is winning at some auction at equilibrium and by switching he will definitely change either the winner
or the price of that auction. Therefore the deviation is not transparent to the rest of the players. However,
the following inequality, that is trivially true when $x_i(v)=1$, will be sufficient to give us the desired bound.
\begin{equation}
(1-\frac{1}{e})v_i x_i(v) \geq(1-\frac{1}{e})v_i y_i(v)
\end{equation}
Combining the last two inequalities we get that for all $i\in \opt(v)$:
\begin{equation}
u_i(b_i',b_{-i}(v_{-i}))\geq (1-\frac{1}{e})v_iy_i(v) - p_{A_v(i)}y_i(v)-(1-\frac{1}{e})v_i x_i(v)
\end{equation}

For $i\notin \opt(v)$ the above inequality trivially holds since the right hand side is non-positive and 
the utility from the deviation can only be non-negative. 

From the equilibrium condition we have:
\begin{equation}
\E_{v_{-i}|v_i}[u_i(b_i(v_i),b_{-i}(v_{-i}))] \geq \E_{v_{-i}|v_i}[(1-\frac{1}{e})v_i y_i(v)-p_{A_v(i)}(v)y_i(v)-(1-\frac{1}{e})v_i x_i(v)]
\end{equation}
Taking expectation over $v_i$, adding for all $i$ and using the fact that $A_v(i)$ is a matching between players in $\opt(v)$
and auctions that took place (and therefore winners of those auctions), we get:
\begin{align*}
\E_{v}[\sum_{i} u_i(b(v))] \geq & \E_{v}[(1-\frac{1}{e})\sum_{i\in \opt(v)}v_i - \sum_{i\in \opt(v)}p_{A_v(i)}(v) - (1-\frac{1}{e})\sum_{i\in \spe(v)} v_i]\\
=& (1-\frac{1}{e})SW_{\opt}-\E_{v}[\sum_{i\in \spe(v)}p_i(v)]-(1-\frac{1}{e})SW_{\spe}
\end{align*}
Rearranging we get:
\begin{align*}
\E_{v}[\sum_{i\in \spe(v)}u_i(b(v))+p_i(v)]+(1-\frac{1}{e})SW_{\spe} \geq(1-\frac{1}{e})SW_{\opt} \implies\\
SW_{\spe} +(1-\frac{1}{e})SW_{\spe} \geq (1-\frac{1}{e})SW_{\opt} \implies (1+\frac{e}{e-1})SW_{\spe}\geq SW_{\opt}
\end{align*}
\end{proof}

\subsection{Single Value Matching Markets}
In this part we consider a special case of our matching market setting where each player has a
single private value $v_i$ but participates in a subset of auctions $S_i$. The set of auctions where each
player participates is common knowledge but the value is private. We call such a setting a Single Value
Matching Market. For this special case of matching markets
we show that the bound of $1+\frac{e}{e-1}$ on the price of anarchy extends even when the values of the bidders
are arbitrarily correlated and not independent.

In terms of an auction setting the latter is a special case of a matroid setting where the matroid is the matchable set of nodes
on the one side of a bipartite graph: Let $G$ be a bipartite graph with two sides $A$ and $B$. This graph gives rise to a matroid on $A$, where a subset $S \subset A$ is independent if there is a matching of $G$ that matches all nodes in $A$. We can think of the nodes of $B$ as service stations and $A$ are the players who want to connect to some station. This example can be thought of both as a matroid auction, but it is also a special case of a matching auction where the graph represents which items can be assigned to each player, and all players value all items that they can get identically. It is interesting to note the difference between the sequential cut auction when applied to this matroid and the sequential item auction. The matching auction, runs auctions selling off each element $B$ individually, and greedily assigns the element auctioned to the winner. (With identical valuations winners will not bid on future items). In contrast the cut auction greedily commits to serve the winner $w$, but does not commit to which node in $B$ to assign it to. We can think of the mechanism as maintaining a tentative matching.

Therefore, if we want to claim a bound on the sequential item auction for this setting we cannot simply apply the bound of
the sequential cut auction from the previous section. However, using a very similar proof with the same type of deviation (bid $v_i/2$
at all the item-auctions you participate in until you get allocated) leads to a proof for the item auction too. Moreover,
such a proof holds for correlated bidders.

\begin{theorem}
The Bayes-Nash Price of anarchy of the Sequential First Price Item Auction in a single-value matching market is at most $1+\frac{e}{e-1}\approx 2.58$ even if the values of the different players are arbitrarily correlated.
\end{theorem} 

%% file: appendixA.tex
\section*{Appendix A: Proof of Theorem \ref{lm:example}}

In this section we present in detail the equilibrium computation of the sequential item auction example of
two items $1$ and $2$ and three bidders $a,b,c$.

Consider the first auction. Since at that point both players face a symmetric setting, it is natural to look at
equilibria where $a,b$ bid symmetrically in the first auction. Thus suppose
that at the first auction both $a$ and $b$ bid according to some strategy monotone $b(v)$. Based on such a symmetric play in the first auction, we will work out analytically the equilibrium and utilities in the second auction, and then solve the resulting differential equation that equilibrium in the first auction has to satisfy.

In the second auction, player $c$ is playing with the loser of the first auction.
Without loss of generality assume this is player $a$. From the fact that player $a$ lost the first auction
and the bidders use monotone strategies, it becomes common knowledge that $a$ has a value less than $b^{-1}(p_1)$, where
$p_1$ is the winning price of the first cut auction. Thus the auction for the second item
is a first price auction with asymmetric bidders: player $a$ is uniformly distributed
in $[0,b^{-1}(p_1)]$ and player $c$ is uniformly distributed in $[0,1]$. The latter is
a well known setting and the bidding strategies can be analytically computed (see e.g.
\cite{Krishna2009}):
\begin{align*}
b_a(v_a)= & \frac{1}{k v_a}(1-\sqrt{1-k v_a^2})=\frac{v_a}{1+\sqrt{1-k v_a^2}}\\
b_c(v_c)= & \frac{1}{k v_c}(\sqrt{1+k v_c^2}-1)=\frac{v_c}{1+\sqrt{1+k v_c^2}}
\end{align*}
where $k = \frac{1}{(b^{-1}(p_1))^2}-1\geq 0$. Observe that if the "weak" player $a$ has value $v_a$ then for the strong player to win he must have value $v_c\geq b_c^{-1}(b_a(v_a))= \frac{v_a}{\sqrt{1-kv_a^2}}\geq v_a$. Hence, 
the allocation in this auction will not be efficient.

This bidding strategies allow us to explicitly compute player's $a$ expected utility from the last auction given
his value $v_a$, and given the price the first auction was won at. (Note that in the first auction
player $a$ has different utility for each possible price outcome even when he looses.)
Let $u_a^2(v_a,p_1)$ denote player $a$'s expected utility from the last auction given his value $v_a$, and a
price announcement of $p_1$ in the first auction won by player $b$.
\begin{align*}
u_a^2(v_a,p_1) = & F_c(b_c^{-1}(b_a(v_a)))(v_a-b_a(v_a))=b_c^{-1}(b_a(v_a))(v_a-b_a(v_a))\\
=&\frac{v_a}{\sqrt{1-kv_a^2}}(v_a-b_a(v_a))=
\frac{v_a^2}{1+\sqrt{1-k v_a^2}}
\end{align*}

Since we assume a symmetric strategy in the first auction, the expected utility of
player $a$ in the first auction if he makes a bid as if he has value $x$ is:
\begin{align*}
u_a^1(x) =~& F_b(x)(v_a-b(x))+(1-F_b(x))\E_{v_b|v_b>x}[u_a^2(v_a,b(v_b))]\\=~&
x(v_a-b(x))+\int_{x}^1 u_a^2(v_a,b(t)) dt
\end{align*}
The first part on the right hand side is the utility of the player from winning the auction and the
second part is the utility of the player when he loses.

For $b(x)$ to be an equilibrium it has to be that the latter utility is maximized at
$v_a$. Thus taking the first order conditions gives us:
\begin{align*}
\frac{\partial u_a^1(x)}{\partial x}\bigg|_{x=v_a}=0 \Leftrightarrow
v_a-b(v_a)-v_ab'(v_a)-u_a^2(v_a,b(v_a))=0
\end{align*}
We also impose the boundary condition that a player with zero value will bid zero $b(0)=0$.
It may be useful to observe that without the final part of $u_a^2(v_a,b(v_a))$ the last first order differential equation
gives as a unique solution the well known first price equilibrium for uniformly distributed bidders
of $b(v_a)=v_a/2$. In our case the latter part was introduced due to the externality
introduced by the subsequent auction and the fact that the utility in the subsequent auction depends
on price at the current one. The solution of the above differential equation gives us the equilibrium with the second auction.

Now, we analytically compute the equilibrium function at the first auction. First we observe
that for $p_1=b(v_a)$ we have: $k = \frac{1}{(b^{-1}(b(v_a)))^2}-1=\frac{1}{v_a^2}-1$.
Hence:
\begin{align*}
u_a^2(v_a,b(v_a))=&\frac{v_a^2}{1+v_a}
\end{align*}
Thus the differential equation for the first auction becomes:
\begin{align*}
v_a-b(v_a)-v_ab'(v_a)-\frac{v_a^2}{1+v_a}=0
\end{align*}
This has a unique solution of $b(v) = 1-\frac{\ln(1+v)}{v}$.  So far we assumed $a$ and $b$ in the first auction
bid in the range of $[0,b^{-1}(1)]$ possibly pretending to have a different value, but never bidding outside this range (above $b^{-1}(1)]$, and that beliefs after he first auction about the loser are updated to a uniform distribution
on $[0,b^{-1}(p)]$.

To complete the equilibrium computation we need to show that this equilibrium arises as a limit of equilibria where players bit uniformly random $U(0,1)$ with a small probability $\epsilon$ and bid
according to the equilibrium $b(x)$ with the remaining. For such strategies the Bayesian update
of beliefs is well defined for any possible bid. In the equilibrium, the updated beliefs have to be limits of
such Bayesian updates as $\epsilon$ goes to $0$. For a price announcement in the range of $b(x)$ this
updates are the classic updates we defined above. For a price announcement above the range of $b(x)$
there two cases: either the loser bid noisily and lost in which case there is
no information revealed or he bid according to the equilibrium and lost in which case again
no information is leaked about his value. Thus the updated belief for any $\epsilon$ will
be $U(0,1)$ for the loser. Thus the limit will also be $U(0,1)$. So for price announcements
that are not predicted by the equilibrium function the updated beliefs are as if there was a price
announcement of $b(1)$. The latter detail completes the Sequential Equilibrium computation.